\documentclass[journal]{IEEEtran}
\usepackage{amsmath}
\usepackage{amssymb}
\usepackage{amsthm}
\usepackage{isomath}
\usepackage{caption}
\usepackage[dvipsnames]{xcolor}
\usepackage{enumitem}
\usepackage{hyperref}
\usepackage{cite}
\usepackage{graphicx}
\usepackage{makecell}

\DeclareMathOperator{\E}{\mathbb{E}}
\DeclareMathOperator{\Var}{\mathrm{Var}}
\newtheorem{lemma}{Lemma}

\newtheorem{proposition}{Proposition}

\hypersetup{
	colorlinks,
	linkcolor={red!50!black},
	citecolor={blue!50!black},
	urlcolor={blue!80!black}
}

\newcommand{\matr}[1]{\boldsymbol{#1}}

\newcommand{\bigzero}{\mathbf{0}}

\newcommand{\paren}[1]{\left(#1\right)} 

\author{Olle Abrahamsson, Danyo Danev and Erik G. Larsson 
	\thanks{The authors are with the Department of Electrical Engineering (ISY), Link{\"o}ping University, 58183 Link{\"o}ping, Sweden  (e-mail:  \{olle.abrahamsson, danyo.danev, erik.g.larsson\}@liu.se).}
	\thanks{Parts of the results in this paper were presented at the 2019 Asilomar SSC conference \cite{abrahamsson2019}.}
}

\title{On the Impact of Random Actions \\on Opinion Dynamics}

\begin{document}
	\maketitle
	\begin{abstract}
		We study opinion dynamics in a social network with stubborn agents who influence their neighbors but who themselves always stick to their initial opinion. We consider first the well-known DeGroot model. While it is known in the literature that this model can lead to consensus even in the presence of a stubborn agent, we show that the same result holds under weaker assumptions than has been previously reported. We then consider a recent extension of the DeGroot model in which the opinion of each agent is a random Bernoulli distributed variable, and by leveraging on the first result we establish that this model also leads to consensus, in the sense of convergence in probability, in the presence of a stubborn agent. Moreover, all agents' opinions converge to that of the stubborn agent. We also consider a variation on this model where the stubborn agent is replaced with a drifting agent and show that herding is achieved also in this case. Finally, we offer a detailed critique of a proof regarding a claim about a closely related model in the recent literature.
	\end{abstract}
	\section{Introduction}
	The study of opinion dynamics in social networks goes back several decades; for a review, see e.g. \cite{proskurnikov2017,proskurnikov2018}. For an overview of recent publications, see e.g. \cite{noorazar2020}. One of the most well-known models is the DeGroot model \cite{degroot} which has been studied extensively (for a literature survey see for instance \cite[Section 3]{proskurnikov2017} and \cite[Section 3]{proskurnikov2018}). In this model an agent's opinion is represented by a continuous real variable, which at each time step \(n\in\{1,2,\dots\}\) is updated to a linear combination of the opinions of itself and its neighbors,
	\begin{equation}\label{eq:degroot_update}
	\matr{x}[n+1] = \matr{T}\matr{x}[n],
	\end{equation}
	where \(\matr{x}[n]\) represents the agents' opinions at time \(n\) and \(\matr{T}\) is a matrix that encodes the trust between agents (this is explained in detail in Section \ref{sec:DeGroot}).
	
	A particular case in opinion dynamics is where one or more agents are stubborn (agents whose opinions remain unchanged independent of the others' opinions). This scenario was first introduced by Mobilia in 2003 \cite{mobilia2003} who established convergence rates towards consensus under the so-called voter model \cite{sudbury73} with a single stubborn agent. The voter model was again considered in \cite{yildiz} where the optimal placement of stubborn agents for maximal influence on the long-term expected opinions was investigated, among other properties. In \cite{ghaderi} the authors considered a model in which agents can have a continuous degree of stubbornness, and gave bounds on the rate of convergence to a consensus of opinions. A more recent study \cite{wai} showed that the influence of stubborn agents under the DeGroot model can, under suitable conditions, be used to recover the topology of the network. Specifically the authors derived equations for the expected opinions of the ordinary (non-stubborn) agents that depend on the topology, and then showed how a regression problem could be formulated which estimated matrices with information about the topology by observing opinions that fit the equations. A related study \cite{zhu2020} explored the problem of inferring the graph Laplacian from consensus dynamics, subject to various levels of uncertainties of diffusion rates, observation times and the input signal power. By leveraging on the spectral properties of the observed data and by utilizing tools from graph signal processing, the authors proposed a set of inference algorithms that were demonstrated numerically to outperform the previous state-of-the-art algorithms.   In \cite{zhou2020}, the authors studied the effect of partially stubborn agents on a modified DeGroot model in which an agent takes into account both the opinions of its neighbors as well as those of its neighbors' neighbors. 
	
	Another class of models of particular interest in relation to this paper incorporates randomness, for example in terms of random interactions \cite{acemoglu2010,mukhopadhyay2016}, or as in \cite{hunter2018}, where at each time \(n\) a randomly selected agent communicates a random opinion to its neighbors. The latter model also features the interesting novelty that an agent may grow increasingly stubborn over time. A recent extension of the DeGroot model which incorporates randomness was given in \cite{scaglione}. Under this setting, at every time step \(n\) each agent \(k\) chooses a Bernoulli distributed random action \(A_k[n]~\sim~\mathrm{Bernoulli}(X_k[n])\), and the corresponding update rule is
	\begin{equation}\label{eq:RA_update}
	\matr{X}[n+1] = (1-\alpha)\matr{X}[n] + \alpha\matr{T}\matr{A}[n],
	\end{equation}
	as described further in Section \ref{sec:RA}.\footnote{We use uppercase letters for random variables, e.g., \(\matr{X}[n]\). They are distinguishable from matrices (which are deterministic), e.g., \(\matr{T},\matr{Q}\), since the matrices are not time dependent.} In this model, which we will refer to as the Random Actions model (RA model for short), the probabilities of the actions, rather than the actions/opinions themselves, are updated as a weighted average over the neighbors' actions.
	\section{Contributions}
	Much of the motivation of this work stems from the RA model in \cite{scaglione} and the proof of Theorem 1 therein, in which all agents are claimed to converge almost surely towards a consensus. However, we have not been able to verify all steps in  the proof of that theorem, 
	and we have concerns with some of the arguments. In this paper we give a detailed critique of the proof of \cite[Theorem 1]{scaglione}, 
and rigorously prove a somewhat weaker version of it.
	
	In particular we extend the RA model by the introduction of a stubborn agent and establish that the opinion dynamics converges in probability to a consensus, and furthermore that all agents adopt the stubborn agent's initial opinion. While this result is intuitively expected, the proof entails some non-trivial mathematical techniques. We also show that this result holds even when the stubborn agent is replaced with a so-called drifting agent. The status of the proofs of various modes of convergence under the different RA models are summarized in Table \ref{table:status}; for details, see Section \ref{sec:openproblems}.
	
	As a stepping stone towards the analysis of the aforementioned models we first consider the DeGroot model with a stubborn agent as described in \cite{wai} and show that the convergence results from that paper can be obtained with weakened assumptions on the model. Specifically, instead of assuming that every ordinary agent has a non-zero trust in the stubborn agent, it suffices to assume that at least one ordinary agent has such a trust.  
	
	A conference version of this paper was published in \cite{abrahamsson2019}. The novel contributions in the current paper are the additional analysis of the RA model with a drifting agent and a detailed critique of the proof of \cite[Theorem 1]{scaglione}. 

	\begin{table}
		\caption{Status of proofs of modes of convergence (almost sure, in mean square and in probability) for the Random Actions model, with and without a stubborn/drifting agent, respectively. Recall that convergence almost surely and convergence in mean square both imply convergence in probability.}
		\label{table:status}
		\begin{tabular}{c|c|c|c}
			 & a.s. & m.s. & pr.\\
			 \hline
			 \textbf{RA}& \makecell{\cite{scaglione}, however \\see our critique\\ in Section \ref{sec:critique}}&\makecell{Open\\problem}& \makecell{Would follow from \\\cite{scaglione}, however \\see our critique}\\
			\hline
			\thead{RA with stubborn\\or drifting agent}&\makecell{Open\\problem}&\makecell{Open\\problem}&\makecell{See our proofs of \\propositions \\ \ref{RA-prop} and \ref{RA-drifting-prop}}
		\end{tabular}
	\end{table}
	\section{Models and Definitions}\label{sec:model}
	In all models described in this section, we will consider a directed, weighted, single-component network with \(K\) nodes, where the nodes are interpreted as agents. Before giving the details of the models, let us at this point remind the reader of some definitions. A \textit{sub-(row)-stochastic} matrix is a square, non-negative matrix such that the row sums are less than or equal to \(1\). The word ``row'' will be omitted and implied from hereon. There are two special cases of these matrices: A \textit{stochastic} matrix is a sub-stochastic matrix where all rows sum to \(1\), and a \textit{strictly sub-stochastic} matrix is a sub-stochastic matrix whose row sums are all strictly less than \(1\).
	\subsection{The DeGroot Model with a Stubborn Agent}\label{sec:DeGroot}
	In the DeGroot model \cite{degroot}, at every time step \(n\in\{0,1,\dots\}\), each agent \(k\in\{1,2,\dots,K\}\) observes the opinions of its neighbors, and updates its opinion to a linear combination of its own opinion and those of its neighbors. The update rule is given by \eqref{eq:degroot_update} where \(\matr{x}[0]\in\mathbb{R}^K\) is a column vector representing the initial opinions of the \(K\) agents and \(\matr{T}\) is a \(K\times K\) stochastic matrix representing the trusts between agents. If \(\matr{T}^n\) converges to a limit \(\matr{T}^{\infty}\) as \(n\to\infty\), then consensus is reached and is given by
	\begin{equation}
	\lim_{n \to \infty}\matr{x}[n] = \lim_{n \to \infty}\matr{T}^n\matr{x}[0] = \matr{T}^{\infty}\matr{x}[0].
	\end{equation}
	If the agents are viewed as nodes in a network, then \(\matr{T}\) is interpreted as an adjacency matrix with elements \(t_{ij}\), and we use the convention that \(t_{ij}>0\) represents an edge from \(j\) to \(i\) whose weight is equal to the trust that \(i\) puts in \(j\).
	
	A special case is when one agent is stubborn, that is, an agent who never updates its opinion, corresponding to a node whose only incoming edge is a self-loop. Let the agents' opinions be partitioned into two sets of opinions, \begin{equation}
		\matr{x}_1[n] \quad \text{ and } \quad \matr{y}[n] = (x_{2}[n], x_{3}[n], \dots,x_K[n])^T,
	\end{equation} 
	held respectively by a stubborn agent and \(K-1\) ordinary agents. Then we write
	\begin{equation}\label{eq:partition}
	\matr{x}[n] = 
	\begin{pmatrix}
	x_1[n] \\ \matr{y}[n]
	\end{pmatrix}.
	\end{equation}
	In this case the trust matrix \(\matr{T}\) has the structure
	\begin{equation}\label{eq:degroot_trust}
	\matr{T} = \begin{bmatrix}
	1 & \bigzero_{1\times (K-1)}\\
	\matr{r} & \matr{Q}
	\end{bmatrix},
	\end{equation}
	where the scalar \(1\) represents the stubborn agent, the vector \(\matr{r}\) with dimensions \((K-1)\times 1\) represents the edges from the stubborn to ordinary agents, and the matrix \(\matr{Q}\) represents the edges between ordinary agents. We will assume that all ordinary agents are strongly connected, i.e., \(\matr{Q}\) is irreducible.
	
	\subsection{The RA Model with a Stubborn Agent}\label{sec:RA}
	In the RA model \cite{scaglione}, at every time step each agent \(k~\in~\{1,2,\dots,K\}\) chooses one of two actions, \(0\) or \(1\), and these actions are generated by a Bernoulli random variable \(A_k[n]\) with probability \(X_k[n]\). The update of these probabilities is governed by \eqref{eq:RA_update}, where \(\alpha\in(0,1)\), \(\matr{T}\) is a trust matrix (as defined in Section \ref{sec:DeGroot}), and 
	\begin{equation}
	\matr{A}[n] = (A_1[n], A_2[n], \dots, A_K[n])^T \in \{0,1\}^K
	\end{equation} are the actions with corresponding probabilities
	\begin{equation}
		\matr{X}[n] = (X_1[n], X_2[n], \dots, X_K[n])^T \in[0,1]^K,
	\end{equation} which themselves are stochastic for \(n>0\).
	
	In the case with a stubborn agent we can assume w.l.o.g. that this is agent \(1\) and that it always chooses action \(0\) with probability \(1\) and thus \(1\) with probability \(0\), i.e., \(X_1[n] = 0\) for all \(n\geq 0\). Analogously to \eqref{eq:partition} and \eqref{eq:degroot_trust}, we have
	\begin{equation}\label{eq:RA_stubborn}
	\matr{X}[n] = \begin{pmatrix}
	0 \\ \matr{Y}[n] \end{pmatrix}, \ 
	\matr{T} = \begin{bmatrix}
	1 & \bigzero\\
	\matr{r} & \matr{Q}
	\end{bmatrix}, \
	\matr{A}[n] = \begin{pmatrix}
	A_1[n] \\ \matr{B}[n]
	\end{pmatrix}
	\end{equation}
	where \(\matr{r}\) has dimension \((K-1)\times 1\) and 
	\begin{equation}
		\matr{B}[n] = (A_2[n], A_3[n], \dots, A_K[n])^T \in \{0,1\}^{K-1}.
	\end{equation}	
	Then \(A_1[n] = 0\) with probability \(1\) for all \(n\geq 0\) and the other agents update as in the original RA model. Again we assume that \(\matr{Q}\) is irreducible.
	
	\subsection{The RA Model with a Drifting Agent}\label{sec:RA-drifting}
	In this model we replace the stubborn agent in Section \ref{sec:RA} with a \textit{drifting} agent which chooses action \(1\) with probability \(f[n]\), where \(f:\mathbb{N}\to [0,1]\) is a deterministic function such that \(\lim_{n \to \infty} f[n] = 0\). In relation to \eqref{eq:RA_update}, we have
	\begin{equation}\label{eq:RA_drifting}
	\matr{T} = \begin{bmatrix}
	1 & \bigzero\\
	\matr{r} & \matr{Q}
	\end{bmatrix}, \
	\matr{A}[n] = \begin{pmatrix}
	A_1[n] \\ \matr{B}[n]
	\end{pmatrix},
	\end{equation}
	where
	\begin{equation}
		A_1[n] =
		\begin{cases}
			1, & \text{w.p. } f[n],\\
			0, & \text{w.p. } 1-f[n].
		\end{cases}
	\end{equation}
	As before, \(\matr{Q}\) is assumed to be irreducible.
	
	\section{Results}
	The first proposition establishes the conditions for convergence of the model in Section \ref{sec:DeGroot} with trust matrix \(\matr{T}\) as defined in \eqref{eq:degroot_trust}.
		\begin{proposition}\label{Tconverge}
		If at least one ordinary agent puts a non-zero trust in the stubborn agent, that is, \(t_{i1}>0\) for some \(i>1\), then the limit \(\matr{T}^n, n\to\infty\), exists and has the structure
		\begin{equation*}
		\matr{T}^{\infty} =
		\begin{bmatrix}
		1 & \bigzero \\
		\matr{(I-Q)}^{-1}\matr{r} & \bigzero
		\end{bmatrix}.
		\end{equation*}
		\end{proposition}
	For the proof of Proposition \ref*{Tconverge} we need the following lemma.
	\begin{lemma}\label{lem:spectralradius}
		Let \(\matr{A}\) be an \(M\times M\) irreducible sub-stochastic matrix with at least one row sum being strictly less than \(1\), and let \(\rho(\matr{A})\) be the spectral radius of \(\matr{A}\). It holds that \(\rho(\matr{A})<1\).
	\end{lemma}
	The proofs of all lemmas in this paper are given in the appendix. Note that for a strictly sub-stochastic matrix \(\matr{A}\), we can remove the assumption of irreducibility since it follows directly from Theorem 8.1.22 in \cite{HorJoh} that \(\rho(\matr{A})<1\).
	\begin{proof}[Proof of Proposition \ref{Tconverge}]
		The \(n\)th power of \(\matr{T}\) is
		\begin{equation}\label{eq: tn}
		\matr{T}^n =
		\begin{bmatrix}
		1 & \bigzero \\
		(\matr{I}+\matr{Q}+\matr{Q}^2+\dots +\matr{Q}^{n-1})\matr{r} & \matr{Q}^n
		\end{bmatrix},
		\end{equation}
		where \(\matr{Q}\) is sub-stochastic with at least one row having sum strictly less than \(1\). This is due to the assumption that \(t_{i1}>0\) for some \(i>1\), and since \(\matr{T}\) is stochastic, the \(i\)th row of \(\matr{Q}\) must have sum less than \(1\). Finally, since \(\matr{Q}\) is irreducible, Lemma \ref{lem:spectralradius} applies, and we have \(\rho(\matr{Q})<1\). By Theorem 5.6.12 in \cite{HorJoh}, this implies that
		\begin{equation}\label{Qconverge}
		\lim_{n \to \infty}\matr{Q}^n=\bigzero.
		\end{equation}
		Now, consider
		\begin{equation}\label{partial}
		(\matr{I}-\matr{Q})\sum_{k=0}^{n-1}\matr{Q}^k = \sum_{k=0}^{n-1}\left(\matr{Q}^k - \matr{Q}^{k+1}\right) = \matr{I} - \matr{Q}^{n}.
		\end{equation}
		By \eqref{Qconverge}, the right hand side of \eqref{partial} tends to \(\matr{I}\) in the limit as \(n\to\infty\), and since \(\rho(\matr{Q})<1\) the matrix \(\matr{I}-\matr{Q}\) is invertible.\footnote{To see this, suppose \(\matr{I}-\matr{Q}\) is not invertible. Then there exists a non-zero vector \(\matr{v}\) such that \((\matr{I}-\matr{Q})\matr{v} = \bigzero\), or equivalently \(\matr{Qv} = \matr{v}\), which shows that \(1\) is an eigenvalue of \(\matr{Q}\). But this is impossible since \(\rho(\matr{Q})<1\).}   It follows that
		\begin{equation}
		\lim_{n \to \infty}\sum_{k=0}^{n-1}\matr{Q}^k=(\matr{I}-\matr{Q})^{-1}.
		\end{equation}
	\end{proof}
	With the previously discussed decomposition of \(\matr{x}[n]\) into stubborn and ordinary agents in \eqref{eq:partition}, the opinions of ordinary agents converge as \(n\to\infty\):
	\begin{equation}\label{consensus2}
	\lim_{n \to \infty}\ \matr{y}[n] = (\matr{I}-\matr{Q})^{-1}\matr{r}x_1[0].
	\end{equation}
	
	The second proposition concerns the RA model in Section \ref{sec:RA} with a stubborn agent.
		\begin{proposition}\label{RA-prop}
		The opinion dynamics of \eqref{eq:RA_update} under the restrictions imposed by \eqref{eq:RA_stubborn} leads to herding in the sense of convergence in probability, i.e., for every \(\varepsilon>0\),
		\begin{equation*}
		\lim_{n \to \infty}\mathbb{P}(X_k[n]<1-\varepsilon) = 0, \ \text{for all } k\in\{1,2,\dots,K\}.
		\end{equation*}
		\end{proposition}

The first part of the proof of Proposition \ref{RA-prop} treats the convergence of opinions towards a consensus in the subnetwork induced by the ordinary agents and follows partly the proof of Theorem 1 in \cite{scaglione}, but with some modifications due to the presence of the stubborn agent. The second part shows that the consensus opinion must be equal to that of the stubborn agent. In this part we deviate from \cite{scaglione} in that we show convergence in probability, as opposed to the claimed proof of almost sure convergence therein, which we have been unable to verify. A detailed discussion of the differences will be provided in Section \ref{sec:critique}. We need the following facts for the main proof.
	\begin{lemma}\label{lem:ms-implies-prob-conv}
		If \(\{W[n]\}_{n=0}^\infty\) is a sequence of random variables such that \(W[n]\in [0,1]\) for all \(n\geq 0\), and
		\[\lim_{n \to \infty}\E[W^2[n](1-W[n])^2] = 0,\]
		then for all \(\varepsilon > 0\),
		\[\lim_{n\to\infty}\mathbb{P}(W[n]\leq\varepsilon \ \cup \ W[n]\geq 1-\varepsilon) = 1.\]
	\end{lemma}
	\begin{lemma}\label{lem:conv-prob}
		Consider the update rule in \eqref{eq:RA_update} with 
		\begin{equation}
			\matr{X}[n] = (X_1[n], X_2[n], \dots, X_K[n])^T.
		\end{equation}
	Suppose agent \(i\) puts some trust in agent \(j\) (so that \(t_{ij}>0\)). If \(X_j[n]\xrightarrow{P}0\) and
		\[\lim_{n \to \infty}\E[X_i^2[n](1-X_i[n])^2] = 0,\] 
		then \(X_i[n]\xrightarrow{P}0\) as well.
	\end{lemma}
	\begin{proof}[Proof of Proposition \ref{RA-prop}]
	 Let \(\matr{Y}[n]\), \(\matr{B}[n]\), \(\matr{r}\) and \(\matr{Q}\) be defined as in \eqref{eq:RA_stubborn}. Since the vector \(\matr{r}\) has at least one positive element, \(\matr{Q}\) is sub-stochastic with at least one row sum strictly less than one, so by Lemma \ref{lem:spectralradius} it has a largest eigenvalue \(\lambda~\in~(0,1)\) with corresponding left eigenvector \(\matr{\psi}\), \(\matr{\psi}^T\matr{Q} = \lambda\matr{\psi}^T\). Let \(S[n]~=~\matr{\psi}^T\matr{Y}[n]\). The proof will proceed as follows: First we show that \(S[n]\) is a strict super-martingale that converges in the limit as \(n\to\infty\) to a random variable \(S[\infty]\). Then we show that the conditional variance of the martingale difference sequence \(S[n]-S[n-1]\) converges to zero in the mean square sense. We conclude that all elements in \(\matr{Y[n]}\) converge in probability to the value of the stubborn agent, \(X_1[0] = 0\).
	 
	 We will now show that \(S[n]\) is a strict super-martingale w.r.t. \(\matr{Y}[n]\), that is, \(\E[S[n+1]\mid \matr{Y}[n]] < S[n]\). First, note that by the update rule in \eqref{eq:RA_update},
	\begin{equation}\label{eq:ordinary_update}
	\matr{X}[n+1] = \begin{pmatrix}0 \\ \matr{Y}[n+1] \end{pmatrix} = (1-\alpha)\begin{pmatrix}0 \\ \matr{Y}[n] \end{pmatrix} + \alpha \matr{T}\begin{pmatrix}0 \\ \matr{B}[n] \end{pmatrix}.
	\end{equation}
 	Then we have
	\begin{equation}
	S[n+1] = \matr{\psi}^T\matr{Y}[n+1] = \matr{\psi}^T\big((1-\alpha)\matr{Y}[n] + \alpha\matr{Q}\matr{B}[n]\big),
	\end{equation} 
	and by taking expectations of both sides conditioned on \(\matr{Y}[n]\) we obtain
	\begin{equation}\label{eq:supmg-prop}
	\begin{aligned}
	\E[S[n+1]|\matr{Y}[n]] &= (1-\alpha)\matr{\psi}^T\matr{Y}[n] + \alpha\lambda\matr{\psi}^T\matr{Y}[n]\\
	&= (1-\alpha(1-\lambda))S[n] < S[n],
	\end{aligned}
	\end{equation}
	since \((1-\lambda) \in (0,1)\) and \(\alpha\in(0,1)\). Thus \(S[n]\) is a strict super-martingale, and since \(S[n] \geq 0\) for all \(n\) it follows from the Martingale Convergence Theorem \cite[Theorem 4.2.12]{Durrett} that
	\begin{equation}\label{eq:supmg-as-conv}
	S[n] \xrightarrow{a.s.} S[\infty], \quad n\to\infty
	\end{equation}
	for some random variable \(S[\infty]\).
	
	Consider now the martingale difference sequence 
	\begin{equation}
		\mathrm{\Delta} S[n] = S[n] - S[n-1]
	\end{equation} for \(n>1\). First note that the almost sure convergence of \(S[n]\) in \eqref{eq:supmg-as-conv} implies
	\begin{equation}\label{eq:ds-as-conv}
	\mathrm{\Delta}S[n] \xrightarrow{a.s.} 0, \quad n\to\infty.
	\end{equation}
	Furthermore, \(\matr{Q}\) is irreducible and non-negative, so by the Perron-Frobenius Theorem  \cite[Theorem 8.4.4]{HorJoh} all elements of \(\matr{\psi}\) are positive. Let \(\matr{\psi}\) be normalized so that \(\matr{\psi}^T\mathbf{1} = 1\), where \(\mathbf{1} = (1,1,\dots,1)^T\). Since \(Y_k[n]\in[0,1], \ k=1,2,\dots,K-1\), for all \(n\geq 0\) we then have \(0 \leq S[n] \leq1\) and 
	\begin{equation}
		\lvert \mathrm{\Delta}S[n]\rvert = \lvert \matr{\psi}^T(\matr{Y}[n]-\matr{Y}[n-1])\rvert \leq 1.
	\end{equation}
	Therefore, by the Dominated Convergence Theorem \cite[Theorem 1.5.8]{Durrett} together with the almost sure convergence in \eqref{eq:ds-as-conv}, \(\mathrm{\Delta}S[n]\) converges to \(0\) in \(m\)th mean, i.e.,
	\begin{equation}\label{m.s.con}
	\lim_{n \to \infty}\E[\lvert\mathrm{\Delta} S[n]\rvert^m] = 0, \ \text{for all } m\geq 1.
	\end{equation}
	We will now show that the variance of \(\mathrm{\Delta}S[n+1]\) conditioned on \(\matr{Y}[n]\) converges to zero in mean square as \(n\to\infty\), and then conclude that the elements of \(\matr{Y}[n]\) converge in probability to all \(0\)s or all \(1\)s. We have:
	\begin{equation}\label{eq:ds-var-cond}
	\begin{aligned}
	&\Var(\mathrm{\Delta}S[n+1]\mid\matr{Y}[n]) \\
	&= \E\big[\big(\mathrm{\Delta}S[n+1] - \E[\mathrm{\Delta}S[n+1]\mid\matr{Y}[n]]\big)^2 \mid\matr{Y}[n]\big] \\
	&= \E\big[\big(\matr{\psi}^T(\matr{Y}[n+1] - \matr{Y}[n]) - \\
	& \ \ \matr{\psi}^T\E[\matr{Y}[n+1] - \matr{Y}[n]\mid\matr{Y}[n]]\big)^2 \mid\matr{Y}[n]\big] \\
	&= \E\big[\big(\matr{\psi}^T\matr{Y}[n+1] - \matr{\psi}^T\E[\matr{Y}[n+1]\mid\matr{Y}[n]]\big)^2 \mid\matr{Y}[n]\big] \\
	&= \E\big[\big(\matr{\psi}^T((1-\alpha)\matr{Y}[n] + \alpha \matr{Q}\matr{B}[n]) \\
	& \ \ - \matr{\psi}^T((1-\alpha)\matr{Y}[n] + \alpha \matr{Q}\matr{Y}[n])\big)^2 \mid\matr{Y}[n]\big] \\
	&= \E\big[\big(\alpha \matr{\psi}^T\matr{Q}(\matr{B}[n] - \matr{Y}[n])\big)^2 \mid\matr{Y}[n]\big] \\
	&= \alpha^2\lambda^2\matr{\psi}^T\E\big[(\matr{B}[n] - \matr{Y}[n])(\matr{B}[n] - \matr{Y}[n])^T \mid\matr{Y}[n]\big]\matr{\psi},
	\end{aligned}
	\end{equation}
	where in the last step we used that \(\matr{\psi}\) is a left eigenvector to \(\matr{Q}\) with eigenvalue \(\lambda\). The actions \(\matr{B}[n]\sim\textrm{Bernoulli}(\matr{Y}[n])\) are statistically independent conditioned on \(\matr{Y}[n]\), so only the diagonal elements of the covariance matrix \(\E\big[(\matr{B}[n] - \matr{Y}[n])\cdot(\matr{B}[n] - \matr{Y}[n])^T \mid\matr{Y}[n]\big]\) are non-zero. They can be expressed explicitly as
	\begin{equation}
	\begin{aligned}
	&\E\big[B_k^2[n]\mid\matr{Y}[n]\big] - \big(\E\big[B_k[n]\mid\matr{Y}[n]\big]\big)^2\\
	&=Y_k[n] - Y_k^2[n]\\
	&= Y_k[n](1-Y_k[n]), \ \text{for all } k=1,2,\dots,K-1.
	\end{aligned}
	\end{equation}
	Therefore,
	\begin{equation}\label{eq:ds-var-cond2}
	\begin{aligned}
	&\Var(\mathrm{\Delta}S[n+1]\mid\matr{Y}[n]) \\
	&= \alpha^2\lambda^2\sum_{k=1}^{K-1}\psi_k^2Y_k[n](1-Y_k[n]).
	\end{aligned}
	\end{equation}
	To see that the left hand side of \eqref{eq:ds-var-cond2} converges to zero in the mean square sense, consider its square:
	\begin{equation}\label{eq:square-var-ds}
	\begin{aligned}
	&(\Var(\mathrm{\Delta} S[n+1]\mid \matr{Y}[n]))^2 \\
	&= \big(\E[(\mathrm{\Delta} S[n+1])^2\mid \matr{Y}[n]] - (\E[\mathrm{\Delta}S[n+1]\mid\matr{Y}[n]])^2\big)^2 \\
	&= \big(\E[(\mathrm{\Delta} S[n+1])^2\mid \matr{Y}[n]]\big)^2 +\big(\E[\mathrm{\Delta}S[n+1]\mid\matr{Y}[n]]\big)^4 \\
	& \ \ \  - 2\E[(\mathrm{\Delta} S[n+1])^2\mid \matr{Y}[n]]\big(\E[\mathrm{\Delta}S[n+1]\mid\matr{Y}[n]]\big)^2\\
	&\leq \big(\E[(\mathrm{\Delta} S[n+1])^2\mid \matr{Y}[n]]\big)^2 +\big(\E[\mathrm{\Delta}S[n+1]\mid\matr{Y}[n]]\big)^4\\
	&\leq \E[(\mathrm{\Delta} S[n+1])^4\mid \matr{Y}[n]] +\E[(\mathrm{\Delta}S[n+1])^4\mid\matr{Y}[n]] \\
	&= 2\E[(\mathrm{\Delta}S[n+1])^4\mid\matr{Y}[n]],
	\end{aligned}
	\end{equation}
	where the first inequality holds since \((\mathrm{\Delta}S[n+1])^2\) is non-negative, and the second inequality is due to Jensen's inequality \cite[Theorem 1.6.2]{Durrett}. By taking expectations on both sides of \eqref{eq:square-var-ds} and using the result of convergence in \(m\)th mean in \eqref{m.s.con}, we obtain
	\begin{equation}
	\begin{aligned}
	& \lim_{n \to \infty}\E[(\Var(\mathrm{\Delta} S[n+1]\mid \matr{Y}[n]))^2] \\
	&\leq 2 \lim_{n \to \infty} \E\big[\E[(\mathrm{\Delta}S[n+1])^4\mid\matr{Y}[n]]\big] \\
	&= 2 \lim_{n \to \infty}\E[(\mathrm{\Delta}S[n+1])^4] = 0.
	\end{aligned}
	\end{equation}
	As already noted, all elements of \(\matr{\psi}\) are positive which, in view of \eqref{eq:ds-var-cond2} together with the mean square convergence just proved, means that
	\begin{equation}\label{eq:MS-leads-to-productsquared}
	\lim_{n \to \infty} \E[(Y_k[n](1-Y_k[n]))^2] = 0, \ \text{for all } k = 1,2,\dots,K-1.
	\end{equation}
	By Lemma \ref{lem:ms-implies-prob-conv} this implies that for all \(Y_k[n], k=1,2,\dots,K-1\) and for all \(\varepsilon>0\), we have
	\begin{equation}\label{eq:ordinary-conv1}
	\lim_{n\to\infty}\mathbb{P}(Y_k[n]<\varepsilon \ \cup \ Y_k[n]>1-\varepsilon) = 1.
	\end{equation}
	
   Let the set of ordinary agents be denoted by \(\mathcal{O}\), and define \(V_0\) as the subset of ordinary agents who put some trust in the stubborn agent, i.e., \(V_0 = \{i\in\mathcal{O}\colon t_{i1}>0\} \subseteq \mathcal{O}\); let \(V_1 \subseteq \mathcal{O}\setminus V_0\) denote the set of ordinary agents who put some trust in at least one of the agents in \(V_0\), and so on. Then by Lemma \ref{lem:conv-prob} together with \eqref{eq:ordinary-conv1} it follows that the elements in \(\{Y_k[n]\colon k\in V_0\}\) must converge in probability to \(0\). Consequently, the elements in \(\{Y_k[n]\colon k\in V_1\}\) must again converge to \(0\).  Since \(\matr{Q}\) is irreducible there is some index \(P\) such that the union of the disjoint sets \(V_1,V_2,\dots,V_P\) makes up the set of ordinary agents, i.e.,
	\begin{equation}
	\bigcup_{p=1}^P V_p = \mathcal{O}.
	\end{equation}
	By repeating the argument for \(V_2,\dots,V_P\), it therefore follows that all elements in \(\{Y_k[n]\colon k\in\mathcal{O}\}\) must converge in probability to the value of the stubborn agent, \(X_1[n]=0\).	
\end{proof}

	\begin{proposition}\label{RA-drifting-prop}
	The opinion dynamics of \eqref{eq:RA_update} under the restrictions imposed by \eqref{eq:RA_drifting} leads to herding in the sense of convergence in probability, i.e., for every \(\varepsilon>0\),
	\begin{equation*}
		\lim_{n \to \infty}\mathbb{P}(X_k[n]<1-\varepsilon) = 0, \ \text{for all } k\in\{1,2,\dots,K\}.
	\end{equation*}
\end{proposition}
\begin{proof}[Proof of Proposition \ref{RA-drifting-prop}]
	In the proof of Proposition \ref{RA-prop} we showed that the random variables \(Y_k[n]\) satisfy the conditions of Lemma \ref{lem:ms-implies-prob-conv}, and since the network for the drifting agent model has the same structure, that result still holds here. Furthermore, \(X_1[n] = f[n]\), so we  have
	\begin{equation}
		\lim_{n\to\infty} X_1[n]  = \lim_{n\to\infty} f[n] = 0,
	\end{equation}
and thus, trivially, \(X_1[n] \xrightarrow{P} 0\). Note also that at least one non-drifting agent \(j\in \{2,\dots,K\}\) puts some trust in the drifting agent, so by Lemma \ref{lem:conv-prob}, \(X_j[n] ~\xrightarrow{P} 0\). From this point it is clear that we can use the same argument as in the end of the proof of Proposition \ref{RA-prop}.
\end{proof}

\section{A critique of the proof of Theorem 1 in \cite{scaglione}}\label{sec:critique}
In \cite[Theorem 1]{scaglione}, it is claimed that the RA model described by \eqref{eq:RA_update} leads to herding, in the sense that
\begin{equation}\label{eq:scaglioneclaim}
\forall k=1,2,\dots K, \ \mathbb{P} \paren{\lim_{n \to \infty}X_k[n] \in \{0,1\}} = 1,
\end{equation}
and moreover that the limit is identical to all agents.\footnote{With our notation, which will be used throughout this section.} This is a stronger result than our Proposition~\ref{RA-prop}, since it states that all agents will almost surely take identical actions in the limit as time goes to infinity, even in the absence of a stubborn agent. In contrast, we proved that the presence of a stubborn agent leads to herding in the sense of convergence in probability. 

The main steps of the proof of \cite[Theorem 1]{scaglione} are: 
\begin{enumerate}
\item First, the random variable \(Q[n] =\matr{ \pi}^T \matr{X}[n]\) is defined, where \(\matr{\pi}\) is a left eigenvector of \(\matr{T}\) with eigenvalue \(1\), and it is shown that \(Q[n]\) is a martingale with respect to \(\matr{x}[n]\), i.e.,
\begin{equation}
\E[Q[n+1]\mid \matr{X}[n]] = Q[n].
\end{equation}
\item Then the martingale difference sequence 
\begin{equation}
	\mathrm{\Delta}Q[n] = Q[n] - Q[n-1]
\end{equation}
is shown to satisfy \(\mathrm{\Delta}Q[n] \xrightarrow{a.s} 0\), as \(n\to\infty\).
\item The almost sure convergence in 2) is used to show that \(\mathrm{\Delta} Q[n]\) converges in the mean square sense, i.e.,  \(\E [\mathrm{(\Delta}Q[n])^2]\to 0\), as \(n\to\infty\). Up to this point in the proof we have been able to verify all the arguments.
 \end{enumerate}
 
 Our main concern with the proof is the following.
 The proof claims that 
 \begin{quotation}
 	``since for all \(k\), \(\pi_k > 0\), the MS convergence implies that
 	\begin{equation}\label{eq:MS-leads-to-claim}
 		\lim_{n\to\infty}X_k[n](1-X_k[n]) = 0, \quad \forall k.\text{''}
 	\end{equation}
 \end{quotation}
  It is not clear in what sense one should understand the convergence in \eqref{eq:MS-leads-to-claim}.
  Clearly, convergence holds in the mean square sense: 
Similar to \eqref{eq:MS-leads-to-productsquared}, we can show that
	\begin{equation}
		\lim_{n\to\infty}\E[(X_k[n](1-X_k[n]))^2] = 0, \quad \forall k.
	\end{equation}
However, for the convergence in \eqref{eq:MS-leads-to-claim} to be useful for subsequent arguments in the proof in \cite{scaglione},
\eqref{eq:MS-leads-to-claim} must hold almost surely. More specifically,
the convergence is later used (in \cite[Equation (18)]{scaglione}) to argue that
\begin{equation}\label{eq:eitheror-claim}
\lim\limits_{n\to\infty}X_k[n] = 0 \text{ or } \lim\limits_{n\to\infty}X_k[n] = 1.
\end{equation}
But \eqref{eq:MS-leads-to-claim}, interpreted in the sense of mean square convergence, does not  imply convergence
in \eqref{eq:MS-leads-to-claim} almost surely, let alone does it imply \eqref{eq:eitheror-claim}. As a counterexample, consider a distribution which always results in the outcome
	\begin{equation}\label{eq:counterexample}
		X_k[n] = \begin{cases}
			1, & n \text{ odd},\\
			0, & n \text{ even}.
		\end{cases}
	\end{equation}
Then \(X_k[n](1-X_k[n]) = 0\) for all \(n\), but \( \lim\limits_{n\to\infty}X_k[n]\) does not exist, regardless of the mode of convergence.

While the issue just explained constitutes our main point of criticism,
we note in passing that \cite[Equation (16)]{scaglione} as written is inaccurate.
That equation  states that
\begin{equation}
\begin{aligned}
\Var (\mathrm{\Delta} Q[n] &\mid Q[n-1]) = \E[(\mathrm{\Delta} Q[n])^2 \mid Q[n-1]]\\
&= \alpha^2\sum_{k=1}^K \pi_k^2\Var(A_k[n])\\
&= \alpha^2 \sum_{k=1}^K\pi_k^2 X_k[n](1-X_k[n]),
\end{aligned}
\end{equation}
but should read
\begin{equation}
\begin{aligned}
\Var(\mathrm{\Delta} Q[n] &\mid \matr{X}[n-1]) = \E[(\mathrm{\Delta} Q[n])^2 \mid \matr{X}[n-1]]\\
&= \alpha^2\sum_{k=1}^K \pi_k^2\Var(A_k[n-1] \mid X_k[n-1])\\
&= \alpha^2 \sum_{k=1}^K\pi_k^2 X_k[n-1](1-X_k[n-1]).
\end{aligned}
\end{equation}
The corrections that should be applied to \cite[Equation (16)]{scaglione}  are the following:
\begin{itemize}
	\item The conditional variance and conditional expectation on the first line should be with respect to \(\matr{X}[n-1]\), not  \(Q[n-1]\). Otherwise one cannot make use of the definition \(Q[n] = \matr{\pi}^T\matr{X}[n]\) to simplify the expression, and hence the subsequent equality would not hold.
	\item On the second line, the variance should be the conditional variance \(\Var(A_k[n-1] \mid X_k[n-1])\). Note also the time shift, which follows from the previous line.
	\item On the third line, the summand should be\\ \(\pi_k^2X_k[n-1](1-X_k[n-1])\), i.e., once again the time variable should be shifted.
\end{itemize}

\section{Possible extensions, open problems and conclusions}\label{sec:openproblems}
In this paper we have shown that the presence of a stubborn or a drifting agent leads all agents to herd to the same action in the sense of convergence in probability. However, to our knowledge it is an open question whether this result can be strengthened to convergence almost surely, or to convergence in the  mean square sense.  From \eqref{eq:MS-leads-to-productsquared} it is clear that the product \(X_k[n](1-X_k[n])\) converges to zero in mean square, but by the counterexample \eqref{eq:counterexample} this does not necessarily mean that each factor converges to zero in mean square. Note, however, that it is in itself an open question whether the sequence in the counterexample can be the outcome of a probability distribution given the model assumptions: It is conceivable that the constraints that the subnetwork of non-stubborn agents must be irreducible, and that the trust matrix is stochastic, are sufficient to keep actions from oscillating that rapidly between extremes, but we have not been able to prove this.

Further, consider the original RA model without a stubborn or drifting agent. In view of our critique of the proof of \cite[Theorem 1]{scaglione}, it is still an open question whether a network under that model leads to herding almost surely. It is also an open question whether the model leads to herding in the mean square sense, or even in probability. Note that this is an independent question from the one about almost sure convergence, since these two notions of convergence do not in general imply each other.

Finally, we remark that the RA model without a stubborn agent behaves chaotically in simulations. For example, all agents might come very close towards a consensus to, say, \(0\) only to collectively jump in a few time steps towards \(1\). It is therefore hard to draw conclusions about convergence to a consensus based on observations of simulated realizations.

We conclude by summarizing the open problems:
\begin{itemize}
	\item Does the RA model lead to herding with almost sure convergence?
	\item Can the correctness of the proof of \cite[Theorem 1]{scaglione} be verified in view of our critique?
	\item Does the RA model lead to herding with mean square convergence?
	\item Can the sequence in \eqref{eq:counterexample} be the outcome of a distribution that is compatible with the stochasticity and irreducability  constraints imposed on the matrices \(\matr{Q}\) and \(\matr{T}\) (defined in \eqref{eq:degroot_trust}), respectively?
	\item Does the RA model with a stubborn/drifting agent lead to herding with almost sure convergence?
	\item Does the RA model with a stubborn/drifting agent lead to herding with mean square convergence?
\end{itemize}
While some of these questions would be settled if \cite[Theorem 1]{scaglione} holds, our critique in Section \ref{sec:critique} casts doubts on the validity of the proof.
\appendix
\begin{proof}[Proof of Lemma \ref{lem:spectralradius}]
	Let \(\mathbf{1} = (1,1,\dots,1)^T\) and for any \(m, \ 1\leq m\leq M,\) let \(r_m^{(n)} = [\matr{A}^n\mathbf{1}]_m\) be the \(m\)-th row sum of \(\matr{A}^n = \{a_{ij}^{(n)}\}\). Since \(\matr{A}\) is sub-stochastic we have that \(0\leq r_m^{(1)} \leq 1\) for all \(m\), and further that for any \(n\geq1\),
	\begin{equation}\label{eq:recursion}
	\begin{aligned}
		&r_m^{(n+1)} = \sum_{j=1}^{M}a_{mj}^{(n+1)} = \sum_{j=1}^{M}\left(\sum_{k=1}^{M}a_{mk}^{(n)}a_{kj}\right)\\
		&=\sum_{k=1}^{M}\left(a_{mk}^{(n)}\sum_{j=1}^{M}a_{kj}\right) = 
		\sum_{k=1}^{M}a_{mk}^{(n)}r_k^{(1)}.
	\end{aligned}
	\end{equation}
	Therefore
		\begin{equation}\label{eq:non-incr}
	\begin{aligned}
	&r_m^{(n+1)} \leq \sum_{k=1}^{M}a_{mk}^{(n)} = r_m^{(n)},
	\end{aligned}
	\end{equation}
	so the row sums are non-increasing with powers of \(\matr{A}\). By assumption at least one row sum is strictly less than \(1\), so w.l.o.g. we can assume that the rows of \(\matr{A}\)  are ordered such that this applies to the first row sum, i.e., \(r_1^{(1)} < 1\). By the irreducibility of \(\matr{A}\), for any \(m\) there is a positive integer \(l_m\) such that \(a_{m1}^{(l_m)}>0\) (since the induced network is strongly connected). In fact, if \(m\not= 1\) we have \(l_m<M\) (take the shortest path from node \(m\) to node \(1\)). By using \eqref{eq:recursion} we therefore obtain, for any row \(m\),
	\begin{equation}
	\begin{aligned}
	&r_m^{(l_m+1)} = \sum_{j=1}^{M}a_{mj}^{(l_m)}r_j^{(1)} = \sum_{j=2}^{M}a_{mj}^{(l_m)}r_j^{(1)} + a_{m1}^{(l_m)}r_1^{(1)} \\
	&\leq \sum_{j=2}^{M}a_{mj}^{(l_m)} + a_{m1}^{(l_m)}r_1^{(1)} < \sum_{j=1}^{M}a_{mj}^{(l_m)} = r_m^{(l_m)},
	\end{aligned}
	\end{equation}
	which together with \eqref{eq:non-incr} shows that every row sum of \(\matr{A}^n\) is strictly less than \(1\) for all \(n\geq M\). By Theorem 8.1.22 in \cite{HorJoh}, the spectral radius of a non-negative matrix is bounded from above by the maximum row sum.
	This means that \(\rho(\matr{A}^M) <1\), and since \(\rho(\matr{A}^M) = \rho(\matr{A})^M\), we therefore obtain \(\rho(\matr{A})<1\).
\end{proof}
\begin{proof}[Proof of Lemma \ref{lem:ms-implies-prob-conv}]
	Let \(\mu > 0\), and set \(\gamma = \mu\varepsilon^2(1-\varepsilon)^2\). We know that
	\begin{equation}
	\begin{aligned}
	&\lim_{n \to \infty}\E[W^2[n] (1-W[n])^2]\\ 
	= &\lim_{n \to \infty}\int_0^1 w^2 (1-w)^2 f_{W[n]}(w)\mathrm{d}w = 0,
	\end{aligned}
	\end{equation}
	where \(f_{W[n]}(w)\) is the probability density function of \(W[n]\). Thus there exists \(N>0\) such that  \(A+B+C < \gamma\) for \(n\geq N\), where
	\begin{equation}
	\begin{aligned}
	&A=\int_0^\varepsilon w^2(1-w)^2 f_{W[n]}(w) \mathrm{d}w,\\
	&B=\int_\varepsilon^{1-\varepsilon} w^2 (1-w)^2 f_{W[n]}(w) \mathrm{d}w,\\
	&C=\int_{1-\varepsilon}^1 w^2(1-w)^2 f_{W[n]}(w) \mathrm{d}w.
	\end{aligned}
	\end{equation}
	But \(A>0\) and \(C>0\), so \(B< \gamma\) for all \(n\geq N\), and 
	\begin{equation}
	\gamma > B \geq \varepsilon^2 (1-\varepsilon)^2  \int_\varepsilon^{1-\varepsilon} f_{W[n]}(w)\mathrm{d}w,
	\end{equation}
	which implies
	\begin{equation}\label{eq:stay-out-of-middle-conv}
	\int_\varepsilon^{1-\varepsilon} f_{W[n]}(w)\mathrm{d}w \leq \frac{\gamma}{\varepsilon^2(1-\varepsilon)^2} = \mu.
	\end{equation}
	Since \eqref{eq:stay-out-of-middle-conv} holds for all \(\mu > 0\) and \(\varepsilon > 0\), we have
	\begin{equation}
	\lim_{n\to\infty}\mathbb{P}(\varepsilon < W[n] < 1 -\varepsilon) = 0,
	\end{equation}
	or equivalently,	
	\begin{equation}\label{eq:ordinary_conv2}
	\lim_{n\to\infty}\mathbb{P}(W[n]\leq\varepsilon \ \cup \ W[n]\geq 1-\varepsilon) = 1.
	\end{equation}
\end{proof}
	\begin{proof}[Proof of Lemma \ref{lem:conv-prob}]
		We know from Lemma \ref{lem:ms-implies-prob-conv} that for all \(\varepsilon > 0\) and \(\delta>0\) there exists \(N_1>0\) such that for all \(n\geq N_1\),
		\begin{equation}\label{eq:yi-prob-zero-or-one}
		\mathbb{P}(\varepsilon < X_i[n+1] < 1-\varepsilon) < \dfrac{\delta}{2}.
		\end{equation}
		The assumption that \(X_j[n] \xrightarrow{P}0\) as \(n\to\infty\), together with the uniform integrability of \(X_j[n]\) (it is bounded by the interval \([0,1]\)) implies that the expected value of \(X_j[n]\) also converges to \(0\). (This is a standard result in probability theory. See, e.g., \cite[Theorem 5.5.2]{Durrett}.) Thus, for all \(\delta>0\) there exists \(N_2>0\) such that for all \(n\geq N_2\),
		\begin{equation}\label{eq:yj-conv-mean}
		\E[X_j[n]] < \dfrac{\delta}{2}.
		\end{equation}
		We want to show that \(X_i[n]\xrightarrow{P}0\) as \(n\to\infty\). To this end, recall that \(\alpha\in(0,1)\) and that \(t_{ij}>0\) since we assume that \(i\) puts a trust in \(j\). Let \(0<\varepsilon<\alpha t_{ij}\) and \(\delta > 0\). Then for all \(n>\max\{N_1,N_2\}\), we have
		\begin{equation}
		\begin{aligned}
		&\mathbb{P}(X_i[n+1] > \varepsilon)\\
		&=\mathbb{P}(\varepsilon < X_i[n+1] < 1-\varepsilon) +  \mathbb{P}(X_i[n+1] \geq 1-\varepsilon)  \\
		&< \dfrac{\delta}{2} + \mathbb{P}(X_i[n+1] \geq 1-\varepsilon)\\
		&= \dfrac{\delta}{2} + \int_{\matr{x}}\mathbb{P}(X_i[n+1]\geq 1 - \varepsilon\mid \matr{X}[n] = \matr{x})f_{\matr{X}[n]}(\matr{x})d\matr{x} \\
		&=\dfrac{\delta}{2} + \int_{\matr{x}}\mathbb{P}\big((1\!-\!\alpha)x_i + \!\alpha\sum_{k=1}^{K}t_{ik}A_k[n]\geq 1 \!- \!\varepsilon\mid \matr{X}[n] = \matr{x}\big)\\
		&\qquad\cdot f_{\matr{X}[n]}(\matr{x})d\matr{x} \\
		&\leq \dfrac{\delta}{2} + \int_{\matr{x}}\mathbb{P}(1\!-\!\alpha + \alpha (1\!-\!t_{ij}(1\!-\!A_j[n]))\geq 1 \!- \!\varepsilon\mid \matr{X}[n] = \matr{x})\\
		&\qquad \cdot f_{\matr{X}[n]}(\matr{x})d\matr{x} \\
		&= \dfrac{\delta}{2} +  \int_{\matr{x}}\mathbb{P}(\alpha t_{ij}(1-A_j[n])\leq \varepsilon\mid \matr{X}[n] = \matr{x}) f_{\matr{X}[n]}(\matr{x})d\matr{x} \\
		&=\dfrac{\delta}{2} +  \int_{\matr{x}}\mathbb{P}(1-A_j[n]\leq\dfrac{\varepsilon}{\alpha t_{ij}}\mid \matr{X}[n] = \matr{x}) f_{\matr{X}[n]}(\matr{x})d\matr{x} \\
		&= \dfrac{\delta}{2} +  \int_{\matr{x}}\mathbb{P}(A_j[n]\geq 1-\dfrac{\varepsilon}{\alpha t_{ij}}> 0\mid \matr{X}[n] = \matr{x}) f_{\matr{X}[n]}(\matr{x})d\matr{x} \\
		&\leq \dfrac{\delta}{2} + \int_{\matr{x}}x_j f_{\matr{X}[n]}(\matr{x})d\matr{x} \\
		&=\dfrac{\delta}{2} +  \E[X_j[n]]  < \dfrac{\delta}{2} + \dfrac{\delta}{2} = \delta,
		\end{aligned}
		\end{equation}

		where the first inequality follows from \eqref{eq:yi-prob-zero-or-one}, the second inequality follows from the facts that 
		\begin{equation}
			\sum_{k=1}^{K}t_{ik}A_k[n] \leq \sum_{k=1, k\neq j}^{K}t_{ik} + t_{ij}A_j[n] = 1 - t_{ij} +t_{ij}A_j[n]
		\end{equation} and \(X_i[n]\leq 1\), and the last inequality follows from \eqref{eq:yj-conv-mean}. We have also used the fact that \(A_j[n]\sim\mathrm{Bernoulli}(X_j[n])\) conditioned on \(X_j[n]\).
	\end{proof}
	\bibliographystyle{IEEEtran}
	\bibliography{IEEEabrv,RA-full-paper} 

\begin{thebibliography}{10}
\providecommand{\url}[1]{#1}
\csname url@samestyle\endcsname
\providecommand{\newblock}{\relax}
\providecommand{\bibinfo}[2]{#2}
\providecommand{\BIBentrySTDinterwordspacing}{\spaceskip=0pt\relax}
\providecommand{\BIBentryALTinterwordstretchfactor}{4}
\providecommand{\BIBentryALTinterwordspacing}{\spaceskip=\fontdimen2\font plus
\BIBentryALTinterwordstretchfactor\fontdimen3\font minus
  \fontdimen4\font\relax}
\providecommand{\BIBforeignlanguage}[2]{{%
\expandafter\ifx\csname l@#1\endcsname\relax
\typeout{** WARNING: IEEEtran.bst: No hyphenation pattern has been}%
\typeout{** loaded for the language `#1'. Using the pattern for}%
\typeout{** the default language instead.}%
\else
\language=\csname l@#1\endcsname
\fi
#2}}
\providecommand{\BIBdecl}{\relax}
\BIBdecl

\bibitem{abrahamsson2019}
O.~{Abrahamsson}, D.~{Danev}, and E.~G. {Larsson}, ``Opinion dynamics with
  random actions and a stubborn agent,'' in \emph{53rd Asilomar Conference on
  Signals, Systems, and Computers}, 2019, pp. 1486--1490.

\bibitem{proskurnikov2017}
A.~V. Proskurnikov and R.~Tempo, ``A tutorial on modeling and analysis of
  dynamic social networks: Part {I},'' \emph{Annu. Rev. Control}, vol.~43, pp.
  65--79, 2017.

\bibitem{proskurnikov2018}
------, ``A tutorial on modeling and analysis of dynamic social networks: Part
  {II},'' \emph{Annu. Rev. Control}, vol.~45, pp. 166--190, 2018.

\bibitem{noorazar2020}
H.~Noorazar, ``Recent advances in opinion propagation dynamics: a 2020
  survey,'' \emph{The European Physical Journal Plus}, vol. 135, no.~6, p. 521,
  Jun 2020.

\bibitem{degroot}
M.~H. DeGroot, ``Reaching a consensus,'' \emph{J. Amer. Statist. Assoc.},
  vol.~69, no. 345, pp. 118--121, 1974.

\bibitem{mobilia2003}
M.~Mobilia, ``Does a single zealot affect an infinite group of voters?''
  \emph{Phys. Rev. Lett.}, vol.~91, no.~2, 028701, 2003.

\bibitem{sudbury73}
A.~Sudbury and P.~Clifford, ``A model for spatial conflict,''
  \emph{Biometrika}, vol.~60, no.~3, pp. 581--588, Dec. 1973.

\bibitem{yildiz}
E.~Yildiz, A.~Ozdaglar, D.~Acemoglu, A.~Saberi, and A.~Scaglione, ``Binary
  opinion dynamics with stubborn agents,'' \emph{{ACM} Trans. Econ. Comput.},
  vol.~1, no.~4, Dec. 2013.

\bibitem{ghaderi}
J.~Ghaderi and R.~Srikant, ``Opinion dynamics in social networks with stubborn
  agents,'' \emph{Automatica}, vol.~50, no.~12, pp. 3209--3215, Dec. 2014.

\bibitem{wai}
H.~{Wai}, A.~{Scaglione}, and A.~{Leshem}, ``Active sensing of social
  networks,'' \emph{{IEEE} Trans. Signal Inf. Process. Netw.}, vol.~2, no.~3,
  pp. 406--419, Sep. 2016.

\bibitem{zhu2020}
Y.~{Zhu}, M.~T. {Schaub}, A.~{Jadbabaie}, and S.~{Segarra}, ``Network inference
  from consensus dynamics with unknown parameters,'' \emph{{IEEE} Trans. Signal
  Inf. Process. Netw.}, vol.~6, pp. 300--315, 2020.

\bibitem{zhou2020}
Q.~Zhou, Z.~Wu, A.~H. Altalhi, and F.~Herrera, ``A two-step communication
  opinion dynamics model with self-persistence and influence index for social
  networks based on the {D}e{G}root model,'' \emph{Information Sciences}, vol.
  519, pp. 363 -- 381, 2020.

\bibitem{acemoglu2010}
D.~Acemoglu, A.~Ozdaglar, and A.~ParandehGheibi, ``Spread of (mis)information
  in social networks,'' \emph{Games Econ. Behav.}, vol.~70, no.~2, pp.
  194--227, 2010.

\bibitem{mukhopadhyay2016}
A.~{Mukhopadhyay}, R.~R. {Mazumdar}, and R.~{Roy}, ``Binary opinion dynamics
  with biased agents and agents with different degrees of stubbornness,'' in
  \emph{Proc. of the 28th International Teletraffic Congress (ITC 28)},
  vol.~01, Sep. 2016, pp. 261--269.

\bibitem{hunter2018}
\BIBentryALTinterwordspacing
D.~S. {Hunter} and T.~{Zaman}, ``Opinion dynamics with stubborn agents,''
  \emph{arXiv e-prints}, Jun. 2018. [Online]. Available:
  \url{https://arxiv.org/abs/1806.11253}
\BIBentrySTDinterwordspacing

\bibitem{scaglione}
A.~{Leshem} and A.~{Scaglione}, ``The impact of random actions on opinion
  dynamics,'' \emph{{IEEE} Trans. Signal Inf. Process. Netw.}, vol.~4, no.~3,
  pp. 576--584, Sep. 2018.

\bibitem{HorJoh}
R.~A. Horn and C.~R. Johnson, \emph{Matrix Analysis}, 2nd~ed.\hskip 1em plus
  0.5em minus 0.4em\relax New York, NY, USA: Cambridge University Press, 2013.

\bibitem{Durrett}
R.~Durrett, \emph{Probability: Theory and Examples}, 2nd~ed.\hskip 1em plus
  0.5em minus 0.4em\relax Cambridge, United Kingdom: Cambridge University
  Press, 2010.

\end{thebibliography}
\end{document}